\documentclass[a4paper]{article}
\usepackage{graphicx} 
\usepackage{amsmath, amsthm, amsfonts, amssymb, amscd, bm, enumerate}
\usepackage{verbatim}

\usepackage{url}

\newtheorem{theorem}{Theorem}

\newtheorem{lemma}{Lemma}

\newtheorem{assumption}{Assumption}

\theoremstyle{remark}
\newtheorem{remark}{Remark}
\newtheorem{example}{Example}

\def\F{\mathcal{F}}

\def\X{\mathcal{X}}
\def\Y{\mathcal{Y}}

\def\bR{\mathbb{R}}

\def\bone{\mathbf{1}}
\def\bP{\mathbb{P}}
\def\bQ{\mathbb{Q}}
\def\bD{\mathbb{D}}
\def\blambda{\boldsymbol\lambda}

\def\diag{\mathrm{diag}}
\mathchardef\mhyphen="2D

\begin{document}
\title{Filters and smoothers for self-exciting Markov modulated counting processes.}
\author{Samuel N. Cohen\thanks{Samuel Cohen acknowledges the research support of the Oxford--Man Institute for Quanititative Finance, in particular with the provision of NYSE Euronext data.} and Robert J. Elliott}
\date{\today}

\maketitle

\begin{abstract}
We consider a self-exciting counting process, the parameters of which depend on a hidden finite-state Markov chain. We derive the optimal filter and smoother for the hidden chain based on observation of the jump process. This filter is in closed form and is finite dimensional. We demonstrate the performance of this filter both with simulated data, and by analysing the `flash crash' of 6th May 2010 in this framework.


\end{abstract}

\section{Introduction}
In many situations, one wishes to infer properties of a hidden state process from noisy observations. In this paper, we focus on a praticular scenario, where the observation process is a counting process (that is, an integer valued jump process which increases by at most one at any time), but one which has the ability to `self-excite', that is, where the rate of jumps can depend on the jumps in the past. We assume that the precise relationship between the past jumps and the current rate is modulated by another process, which is unobservable, but is a finite-state continuous time Markov chain. Our challenge is to detect the state of this underlying process from our observations.

Problems similar to this are common in many settings. Self-exciting processes have been given increasing importance in recent years as models involved in high-frequency trading on financial markets. Such models can be found in \cite{Engle1998, Koopman08, Bacry11, Cartea2011, Cont11, Chen13} among others, and \cite{Bauwens09} provides a survey. Other examples of applications of self-exciting processes include earthquake occurrence \cite{Ogata88, VereJones95}, neuron firing \cite{Chornoboy88} and criminal activity \cite{Mohler11}. 

On the other hand, hidden Markov models have been used to model a wide variety of economic, financial and industrial phenomena. Examples of this form the focus of the book \cite{Elliott1994a}. The filters for these processes, using the reference probability due to Zakai, were derived in \cite{Elliott1993}, for the case where the observation process is an Ito diffusion. Closer to our situation here is \cite{Elliott2005a}, for the case when observations follow a Poisson process with rate dependent on the underlying state. A key advantage of these methods is that, due to the assumption that the underlying process is a finite state Markov chain, the estimate of the underlying state is obtainable in a closed form, up to the solution of a deterministic ODE. This gives significant computational advantages, as these filtering equations are extremely fast to calculate.

Combining filtering techniques with high-frequency financial data is a well established principle. Examples of this, using various models, can be found in \cite{Bolland97, Frey01, Platania04} and many others. In particular Frey and Runggaldier \cite{Frey01} have a very similar situation to ours, with a self-exciting Markov-modulated jump process, however their filters are different (they mention but do not explore the case when $X$ is a finite state Markov chain), and they do not consider application of these methods to any specific models or datasets.

The paper is structured as follows. In Section \ref{sec:model}, we introduce the general model under consideration, and also the specific case related to the Hawkes' process. In Section \ref{sec:filter}, extending the Markov-modulated Poisson observation case of \cite{Elliott2005a} we derive the filter and smoother equations, along with simple numerical approximations. We also derive a robust version of the filter, and indicate why this is of limited utility in this context.  In Section \ref{sec:estimation} we discuss the estimation of parameters in this context, including showing why the commonly suggested EM algorithm fails to estimate the transition matrix of the underlying chain. In Section \ref{sec:simulated} we apply these methods to simulated data, and in Section \ref{sec:flash} we apply these methods to TAQ data surrounding the `flash crash' of 6th May 2010.

\section{A self-exciting hidden Markov model}\label{sec:model}
We assume that we have the following situation. Let $X$ be an $N$-state Markov chain, with rate matrix\footnote{Depending on convention, this is either the rate matrix or the transpose of the rate matrix. In our setting, the element $[A_t]{ij}$ is the rate of jumping from state $j$ to state $i$ at time $t$, and so $A^{\top}_t$ is the infinitesimal generator of the process.} $A_t$. We assume, without loss of generality, that $X$ takes values in the standard basis vectors in $\bR^N$, which we denote $\X$ for convenience. Then, as in \cite{Elliott1993}, $X$ has the representation
\[X_t= X_0+ \int_{]0,t]} A_t X_t dt + M_t\]
for $M$ an $\bR^N$-valued $\bP$-martingale. 

Let $Y$ be a counting process, that is, a right-constant integer valued increasing process with jumps of at most $1$. Write $\bD$ for the space of paths of $Y$. Define the filtrations $\{\Y_t\}_{t\geq 0}$ and $\{\F_t\}_{t\geq 0}$, where
\[\Y_t =\sigma(Y_s; s\leq t), \qquad  \F_t =\sigma(Y_s, X_s; s\leq t) = \Y_t\vee \sigma(X_s; s\leq t).\]
Let $\mathcal{P}$ denote the predictable $\sigma$-algebra in the $\{\Y_t\}_{t\geq 0}$ filtration.

We suppose that we have a probability measure $\bP$ such that $Y$ is a Cox process with  $\bP$-compensator (in the $\{\F_t\}_{t\geq 0}$ filtration)
\[\lambda:\X\times[0,\infty[\times\bD\to ]0,\infty[\]
which is $B(\X)\times \mathcal{P}$ measurable and such that $s\mapsto \lambda(X_{s-},s, Y_{(\cdot)})$ is left continuous. That is,
\[Y_t - \int_{]0,t]} \lambda( X_{s-},s, Y_{(\cdot)}) ds\] 
is a $(\bP, \{\F_t\}_{t\geq0})$-martingale. The value of $\lambda(X_{s-}, s, Y_{(\cdot)})$ is the rate of jumps of $Y$ at time $t$.  Note that only the `state a moment ago' $X_{t-}$ appears in the rate function, but the previously observed process $Y$ can affect the rate in a general way (however only past observations can have an effect, due to the predictability assumption).

\begin{remark}
 As $Y$ appears in the intensity, it is difficult to show whether $\lambda$ is integrable and $Y_t$ is finite-valued for all $t$ with probability one (that is, we have a density for $Y$ of an appropriate class). For $\lambda$ of the special form
\[\lambda(t, Y_{(\cdot)}) = \phi\Big(\int_{]-\infty,t]} h(t-s) dY_s\Big)\]
appropriate conditions for stability are given in \cite{Bremaud1996}. Rather than concern ourselves with this technical problem, we will simply assume that we have a model which satisfies these properties. A formal condition is given below (Assumption \ref{assn:Zmart}).
\end{remark}

\begin{example}
 A key example of interest is when 
\[\lambda(X_{t-},t, Y_{(\cdot)}) = \langle \alpha, X_{t-}\rangle + \langle \beta, X_{t-}\rangle \int_{]0, t[} e^{-\langle \gamma, X_{t-}\rangle (t-s)} dY_s\]
where $\alpha, \beta, \gamma$ are known vectors with nonnegative entries. This is a natural variant of the Hawkes' process mentioned before, with parameters determined by the current state of the Markov chain. Note that the superficially similar situation
\[\lambda = \langle \alpha, X_{t-}\rangle + \int_{]0, t[} \langle \beta, X_s\rangle e^{-\langle \gamma, X_s\rangle (t-s)} dY_s\]
does not fall into our class of models, as past values of the state $X_s$ have an effect on the current rate of jumps. This restriction is needed to ensure that a finite-dimensional filter is obtained.
\end{example}

\section{Filter and Smoother equations}\label{sec:filter}
Our challenge is, given observations of $Y$, to determine the state $X$. In particular, we wish to be able to evaluate $E[f(X_t)|\Y_T]$ for any function $f:\X\to\bR$ and any times $t$ and $T$. As $\X$ is the space of basis vectors, it is easy to see that any function $f:\X\to\bR$ can be written $f(X_t) = \langle \mathbf{f}, X_t\rangle$ for some vector $\mathbf{f}\in\bR^N$, and so $E[f(X_t)|\Y_T] = \langle \mathbf{f}, E[X_t|\Y_T]\rangle$. For this reason, it is enough for us to determine $E[X_t|\Y_T]$, for all times $t,T$. 

For $T<t$, as $X$ is a Markov chain, in the time-homogenous case we have
\[E[X_t|\Y_T] = E[E[X_t|\Y_T\vee X_T]|\Y_T] =E[e^{A(t-T)} X_T|\Y_T] = e^{A(t-T)} E[X_T|\Y_T]\]
so it is sufficient to consider the case $T\geq t$.  If $A$ is not time homogenous, this equation would be slightly different (the appropriate transition matrix would be used in place of $e^{A(t-T)}$), but the same simplification is possible.

\subsection{Filter equation ($T=t$)}

We first seek to determine an equation for $E[X_t|\Y_t]$, we do this in a similar way to \cite{Elliott2005a}. Write $\lambda_u=\lambda(u, X_{u-} Y_{(\cdot)})$ for notational simplicity. Let $\bQ$ be the measure under which $Y$ is a standard Poisson process, independent of $X$,  and $X$ has dynamics as above.
Let $Z$ be the solution to the equation
\begin{align*}
Z_t &= 1+ \int_{]0,t]} Z_{u-} (\lambda_{u}-1)(dY_u-du)\\
&= \exp\Big(-\int_{]0,t]}(\lambda_{u}-1)du\Big) \Big(\prod_{u\in]0,t]} \lambda_u^{\Delta Y_u}\Big).
\end{align*}
Note that as $Y_t-t$ defines a $\bQ$-martingale, we see that $Z$ is a nonnegative $\bQ$-local martingale with $Z_0=1$. 
\begin{assumption}\label{assn:Zmart}
 $Z$ is a true martingale on $[0,T]$.
\end{assumption}
A sufficient condition under which this is true is given by the generalized Novikov condition
\[E\Big[\exp\Big(\int_{]0,T]}(\lambda_{u}-1)du\Big)\Big]<\infty.\]

\begin{lemma}
 If $Z$ is a true martingale and $\bQ$ is the measure under which $Y$ is a standard Poisson process and $X$ is a Markov chain with rate matrix $A_t$ at time $t$, then under the measure $\tilde\bP$ defined by 
\[\frac{d\tilde\bP}{d\bQ}\Big|_{\F_t} = Z_t,\]
$Y$ has compensator $\int_{]0,t]}\lambda_udu$ and $X$ is a Markov chain with rate matrix as before. Therefore, $\tilde\bP=\bP$.
\end{lemma}
\begin{proof}
From Ito's formula we see that
\[\frac{d(Z_tY_t)}{Z_{t-}} = Y_t(\lambda_{t}-1)(dY_t-dt) + \lambda_{u}dY_t\]
and so 
\[d(Z_tY_t)= Z_{t-}\lambda_t dt+(\bQ\text{-martingale increment}),\]
which implies that $Y$ has the compensator $\int_{]0,t]}\lambda_udu$ under $\tilde\bP$, as desired. Similarly
\[Z_t X_t = X_0 + \int_{]0,t]} Z_{u-} X_{u-} (\lambda_{u}-1) (dY_u -du) + \int_{]0,t]} Z_{u-} A_uX_u du + \int_{]0,t]} Z_{u-} dM_u.\]
Therefore $X$ has $\{\F_t\}_{t\geq0}$-compensator $\int_{]0,t]}A_uX_udu$ under $\tilde\bP$, and so is a Markov chain with rate matrix $A_u$. As our space is defined by the paths of $Y$ and $X$, this implies $\tilde\bP=\bP$.
\end{proof}

We wish to determine a recursive relation for
\[E[X_t|\Y_t] = \frac{E_\bQ[Z_tX_t|\Y_t]}{E_\bQ[Z_t|\Y_t]}.\]
Using the Zakai reference probability approach, we will do this by finding the unnormalized density
\[q_t = E_\bQ[Z_tX_t|\Y_t],\]
as $E_\bQ[Z_t|\Y_t]= \langle q_t, \bone\rangle \in \bR$, this will allow us to calculate our expectations. 

\begin{theorem}
 The unnormalized `filtered' density process $q$ satisfies the $\{\Y_t\}_{t\geq0}$-adapted stochastic vector ODE 
\[q_t = q_0 + \int_{]0,t]}(\Lambda_u-I)q_{u-} (dY_u -du)+ \int_{]0,t]} A_uq_{u} du.\]
\end{theorem}
\begin{proof}
We know that
\[Z_t X_t = X_0 + \int_{]0,t} Z_{u-} X_{u-} (\lambda_{u}-1) (dY_u -du) + \int_{]0,t]} Z_{u-} A_uX_u du + \int_{]0,t]} Z_{u-} dM_u.\]
Then, as $Y$ has independent increments under $\bQ$, taking a conditional expectation ,
\[q_t = q_0 + \int_{]0,t]}E_{\bQ}[Z_{u-} X_{u-} (\lambda_{u}-1)|\Y_u] (dY_u -du)+ \int_{]0,t]} A_uq_{u} du.\]
As $\lambda$ is a scalar function of the state $X_{u-}$, we can write
\[X_{u-} \lambda_{u} =  \lambda(u, X_{u-}, Y_{(\cdot)})X_{u-} = \Lambda_u X_{u-}\]
where $\Lambda$ is the diagonal matrix with entries $(\Lambda_u)_{ii} = \lambda(u, e_i, Y_{(\cdot)})$, and so, crucially, $\Lambda_u$ is $\Y_u$ measurable. Substituting this in our equation for $q$, we find
\[q_t = q_0 + \int_{]0,t]}(\Lambda_u-I)q_{u-} (dY_u -du)+ \int_{]0,t]} A_uq_{u} du.\]
\end{proof}
 Due to the discrete nature of the jumps of $Y$, this stochastic system can often be well approximated using classical ODE techniques.

\subsection{Smoother equation ($T>t$)}
We now wish to find a formula for $r_t:=E_{\bQ}[Z_TX_t|\Y_T]\propto E[X_t|\Y_T]$, when $T>t$. 

\begin{theorem}
 For every $t\geq 0$ there exists a $\Y_T$-measurable vector $v_t$ such that the `smoothed' density $r_t$ satisfies
\[E_{\bQ}[Z_TX_t|\Y_T]= r_t = \diag(v_t) q_t.\]
Furthermore, as a function of $t$, $v$ satisfies the equation
\[dv_t = -(A^\top_t - \Lambda_t +I) v_{t-} dt + (\Lambda_t^{-1}-I)v_{t-} dY_t; \qquad v_T = \bone.\]
\end{theorem}
\begin{proof}
First note that 
\[E_{\bQ}[Z_TX_t|\Y_T]= E_{\bQ}[Z_tX_t E_{\bQ}[Z_T/Z_t|\Y_T\vee\F_t]|\Y_T].\]
As $X$ is a Markov chain and
\[Z_T/Z_t = \exp\Big(-\int_{]t,T]}(\lambda_{u}-1)du\Big) \Big(\prod_{u\in]t,T]} \lambda_u^{\Delta Y_u}\Big)\]
does not depend on $X_s$ for $s<t$, we have 
\[E_{\bQ}[Z_T/Z_t|\Y_T\vee\F_t] = E_{\bQ}[Z_T/Z_t|\Y_T\vee\sigma(X_t)] = \langle v_t, X_t\rangle\]
for some family of $\Y_T$-measurable random variables $\{v_t\}_{t\leq T}$. Therefore, we can write
\[\begin{split}
   r_t=E_{\bQ}[Z_TX_t|\Y_T]&= E_{\bQ}[Z_t X_t \langle v_t, X_t\rangle|\Y_T]= E_{\bQ}[Z_t X_t X_t^\top v_t|\Y_T]\\
&= E_{\bQ}[Z_t \diag(v_t)X_t |\Y_T] =\diag(v_t) E_{\bQ}[Z_t X_t |\Y_T]\\&= \diag(v_t) q_t,
  \end{split}
\]
which is a vector with entries $(r_t)_i=(v_t)_i(q_t)_i$. These variables satisfy $v_T = \bone$ and 
\[E_{\bQ}[Z_T|\Y_T]=\langle r_t, \bone\rangle = \langle  v_t, q_t\rangle\]
for all $t$. 

Making the ansatz $dv_t = H_t v_t dt + K_t v_t dY_t$ for some predictable matrix processes $H,K$, we expand using Ito's rule to find
\begin{align*}
 0&=d\langle v_t, q_t\rangle = \langle  dv_t, q_t\rangle + \langle  v_t, dq_t\rangle + \langle \Delta v_t , \Delta q_t\rangle\\
&= \langle H_t v_tdt+K_t v_tdY_t, q_t\rangle  + \langle  v_t, A_tq_t dt + (\Lambda_t-I)q_t (dY_t-dt)\rangle\\
&\qquad + \langle K_t v_t, (\Lambda_t-I) q_t\rangle dY_t\\
&= \Big\langle v_t, \Big(H_t^\top + A_t - (\Lambda_t-I)\Big) q_t \Big\rangle dt\\
&\qquad \Big\langle v_t, \Big(K_t^\top + (\Lambda_t-I)+ K_t^\top(\Lambda_t-I)\Big) q_t \Big\rangle dY_t.
\end{align*}
Hence, we see that $v_t$ must have the dynamics
\[dv_t = -(A^\top - \Lambda_t+I)v_{t-} dt +(\Lambda_t^{-1}-I)v_{t-} dY_t.\]
\end{proof}
\begin{remark}
It is useful to see that, at a time $\tau$ where $Y$ jumps, we have 
\[v_\tau = v_{\tau-} + (\Lambda_\tau^{-1}-I)v_{\tau-} = \Lambda_\tau^{-1} v_{\tau-}\]
and so $v_{\tau-} = \Lambda_{\tau} v_\tau$.
\end{remark}

\begin{remark}
 In both the filter and the smoother, we have derived equations for the unnormalised densities $q$ and $r$. However, practically it is often convenient to normalise `on the fly', to increase numerical stability. As all our equations are linear in $q$ and $v$ (and $r$ is simply a product of $q$ and $v$), normalising by multiplication by a scalar process does not cause any difficulties, and the result is still correct, up to multiplication by a constant.
\end{remark}

\subsection{Concrete Example: Hawkes' process}
We will now consider a concrete example, for which we can derive an explicit filter. In this context, we will also derive a higher-order approximation to the robust filter in discrete time.

Let $\lambda$ be given by
\[\lambda(X_{t-}, t, Y_{(\cdot)}) = \langle \alpha, X_{t-}\rangle + \langle \beta, X_{t-}\rangle \int_{[0, t[} e^{-\langle \gamma, X_{t-}\rangle (t-s)} dY_s\]
for known nonnegative vectors $\alpha, \beta, \gamma$. We consider two key settings, when jumps are observed in continuous time, and when count data is observed on a discrete timegrid.

\subsubsection{Continuous observations}
When jumps are observed in continuous time, is is in principle possible to implement the filter fully. This involves solving the following system of equations
\[\begin{split}
\text{Between jumps:}&\\
 &\begin{cases}
  d \Lambda_t &= -\diag(\gamma) (\Lambda_t-\diag(\alpha)) dt\\
  d q_t &= (A - (\Lambda_t- I) ) q_t dt \\
  d v_t &= -(A^\top - (\Lambda_t- I) ) v_t dt \\
   \end{cases}\\
\text{At jumps:}&\\
 &\begin{cases}
  \Lambda_{t+} &= \Lambda_{t} + \diag(\beta) \\
  q_t &=  \Lambda_{t-} q_{t-}\\
  v_{t-} &= \Lambda_{t-} v_t \\
   \end{cases}\\
\text{Boundary Values:}& \qquad \Lambda_0= \diag(\alpha), \quad  q_0 = X_0, \quad  v_T= \bone
\end{split}\]
(Note that $\Lambda$ is left-continuous, while $q$ and $v$ are cadlag.) The only difficulty in resolving this system is that we have non-autonomous ODEs for $q$ and $v$, and so a numerical integration technique is needed. For example, using an exponential integrator method gives the solution between jumps of $Y$ at ${t_{i-1}}$ and ${t_i}$, 
\[\begin{cases}
  \Lambda_{t_i} &= \diag(\alpha)+ \exp(-\diag(\gamma)(t_i-t_{i-1})) (\Lambda_{t_{i-1}}+\diag(\beta-\alpha))\\
  q_{t_i} &\approx \Lambda_{t_i}\exp\big((A - (\Lambda_{t_{i-1}}- I))(t_i-t_{i-1})\big) q_{t_{i-1}} \\
  v_{t_{i-1}} &\approx \exp\big((A^\top - (\Lambda_{t_{i-1}}- I))(t_i-t_{i-1})\big) \Lambda_{t_i} v_{t_i}\\
   \end{cases}\\
\]
Including further points between the jumps and using these in the solution of the ODEs will also improve performance, depending on the parameters of the problem.

\subsubsection{Discrete observations}
When observations are observed only discretely, that is, we have a fixed time grid $\{t_i\}$ and observe the values $Y_{t_i}$, then an alternative method is appropriate. First note that, 
\[\begin{split}
   \Lambda_{t_i} &=\diag(\alpha) + \exp(-\diag(\gamma)(t_i-t_{i-1}))(\Lambda_{t_{i-1}} -\diag(\alpha)) \\
& \qquad + \diag(\beta)\int_{[t_{i-1}, t_i[}\exp(-\diag(\gamma)(s-t_{i-1})) dY_s
  \end{split}\]
Provided $\gamma(t_i-t_{i-1})$ is not large, $Y$ is approximately a Poisson process with constant rate on the interval $]t_{i-1}, t_i]$, and so if we observe $\delta Y_{t_i}$ jumps in the interval $]t_{i-1}, t_i]$, then these jumps are approximately uniformly distributed within the interval.  Therefore,
\begin{equation}\label{eq:lambdaDTapprox}
 \begin{split}
   \Lambda_{t_i} &\approx \diag(\alpha) + \exp(-\diag(\gamma)(t_i-t_{i-1}))(\Lambda_{t_{i-1}} -\diag(\alpha)) \\
& \qquad + \diag\Big(\beta\, \frac{1-\exp(-\gamma(t_i-t_{i-1}))}{\gamma(t_i-t_{i-1})}\Big) \delta Y_{t_i}
  \end{split}
\end{equation}
where the operations inside the parentheses on the second line are to be read componentwise.
(Making this approximation gives a slight reduction in bias vs a na{\"\i}ve scheme, where the last term is simply $+\diag(\beta)\,\delta Y_{t_i}$.)

As the equations for $q$ and $v$ are exponential in form, an effective approximation is given by 
\begin{equation}\label{eq:filterDTapprox}\begin{split}
     q_{t_i} &\approx \exp\big((A - (\Lambda_{t_{i-1}}- I))(t_i-t_{i-1})+ {\delta Y_{t_i}} \log(\Lambda_{t_{i-1}})\big)\cdot q_{t_{i-1}} \\
  v_{t_{i-1}} &\approx \exp\big((A^\top - (\Lambda_{t_{i-1}}- I))(t_i-t_{i-1})+ {\delta Y_{t_i}} \log(\Lambda_{t_{i-1}})\big)\cdot  v_{t_i}.\\
  \end{split}
\end{equation}
Combining these equations one can approximate the continuous time filter.

\subsection{Robust filters}
For completeness, we now outline a `robust' filter, where we can avoid integrating against $Y$ in the filtering equations. 

\begin{theorem}
Define the matrix-valued process
\[\Gamma_t = \exp\Big(-\int_{]0,t]} (\Lambda_s - I)ds\Big)\cdot \prod_{s\in]0,t]} (\Lambda_s)^{\Delta Y_s}.\]
Then $\bar q_t:=\Gamma^{-1}_t q_t$ and $\bar v_t = \Gamma_t v_t$ satisfy
\[\begin{split}
   \bar q_t &= \bar q_0 + \int_{]0,t]} \Gamma^{-1}_s A_s \Gamma_s \bar q_s ds\\
 \bar v_t &= \bar v_0 - \int_{]0,t]}\Gamma_s A^\top_s \Gamma_s^{-1}\bar v_s ds= \Gamma_T\bone + \int_{]t,T]}\Gamma_s A^\top_s \Gamma_s^{-1}\bar v_s ds
  \end{split}
\]
Note that these equations do not directly involve the observation process $Y$ (as it only appears through $\Gamma$).
\end{theorem}

\begin{proof}
First note that $\Gamma, \Lambda$ are both diagonal matrices, and therefore commute. From the definition, we see that $\Gamma^{-1}_t$ has dynamics
\[d\Gamma_t^{-1} = (\Lambda_t - I)\Gamma_t^{-1} dt + (\Lambda_t^{-1}-I)\Gamma_{t-}^{-1} dY_s.\]
Now by Ito's rule,
\[\begin{split} d(\Gamma_t^{-1} q_t) &= \Gamma_{t-}^{-1}\big(\Lambda_t - I)q_{t-}(dY_t-dt)+ Aq_t dt\big) \\
&\qquad+ \big((\Lambda_t - I)\Gamma_t^{-1} dt + (\Lambda_t^{-1}-1)\Gamma_{t-}^{-1} dY_s\big)q_{t-}\\
   &\qquad +(\Lambda_t^{-1}-1)\Gamma_{t-}^{-1}(\Lambda_t - I)q_{t-} dY_s\\
&=\Gamma^{-1}_t A q_t dt
  \end{split}
\]
so, writing $\bar q_t = \Gamma^{-1}_t q_t$, we have
\[\bar q_t = \bar q_0 + \int_{]0,t]} \Gamma^{-1}_s A \Gamma_s \bar q_s ds.\]

The corresponding smoother can also be derived, by writing
\[\langle q_t, v_t\rangle = \langle \bar q_t, \bar v_t\rangle = \langle \Gamma_t^{-1} q_t, \bar v_t\rangle\]
from which we see $\bar v_t = \Gamma_t v_t$. Applying Ito's rule, this satisfies
\[\bar v_t = \bar v_0 - \int_{]0,t]}\Gamma_s A^\top \Gamma_s^{-1}\bar v_s ds,\]
or,  as $v_T=\bone$,
\[\bar v_t = \Gamma_T\bone + \int_{]t,T]}\Gamma_s A^\top \Gamma_s^{-1}\bar v_s ds.\]
\end{proof}

\begin{remark}Numerically,  these equations have a significant flaw. As $\Gamma$ involves an exponential, unless $\Lambda_s$ is close to $I$, we will typically have at least one of $\Gamma$ and $\Gamma^{-1}$ growing quickly. Due to noncommutativity, this implies that some components of $\Gamma^{-1}_s A \Gamma_s$ will be large. As $\bar q$ grows like $\exp(\Gamma^{-1}_s A \Gamma_s)$, this leads to a double exponential, which rapidly becomes numerically unstable. This limits the practical applicability of these robust filters.
\end{remark}

\section{Estimation of Parameters}\label{sec:estimation}
To implement these filters, it is often necessary to estimate the parameters of the process using a training dataset. We separate our estimation into two parts -- estimation of the parameters determining $\lambda$, and estimation of the transition matrix $A$ of the chain. This distinction is useful as we will typically have a large number of observed jumps in $Y$ relative to the number of jumps in the underlying chain -- indeed, it is this `multiscale' behaviour which allows the filter to work. As the estimation of $A$ depends on the number of jumps of $X$, this implies that we will often be able to estimate the parameters of $\lambda$ very well given $A$, however our estimation of $A$ requires a much larger amount of data (and is correspondingly much more costly).

\subsection{Observation Parameters}
We proceed to find a maximum likelihood estimator given the state path $X$, which we can then combine with the EM algorithm. The following result is standard (see, for example, \cite{Chen13})
\begin{lemma}
The log-likelihood function given $X$ and observations on the interval $[0,T]$ is, up to the addition of a constant,
\[\sum_{u:\Delta Y_u \neq 0} \log\big(\lambda(u, X_u, Y_{(\cdot)})\big) - \int_0^T\big(\lambda(u, X_u, Y_{(\cdot)})\big) du\]
\end{lemma}

This then leads to the partial-information likelihood function, given the probabilties $\hat r_u = E[X_u|\Y_u]\propto r_u$. We write $\blambda_u$ for the vector with entries $\lambda(u, e_i, Y_{(\cdot)})$.
\begin{lemma}
The log-likelihood function given observations on the interval $[0,T]$ is, up to the addition of a constant,
\[\sum_{u:\Delta Y_u \neq 0} \langle \hat r_u, \log(\blambda_u)\rangle - \int_0^T\langle \hat r_u, \blambda_u \rangle du\]
\end{lemma}
This function can then be optimised numerically, or analytically, depending on the properties of $\lambda$. The EM algorithm then allows us to iterate between estimating the parameters of $\blambda$ given $r$, which can be done using numerical optimization techniques, and calculating $r$ given the parameters of $\blambda$, which can be done using the filtering equations.

\begin{remark}
In some cases, it may be difficult to specify good initial values for the parameters of $\blambda$, however, plotting the observation data may reveal natural approximate clusterings by inspection. Using this \emph{ad hoc} clustering as an initial value of $r$ allows us to first calculate the parameters of $\blambda$, and then to implement these via the filter. This technique is used below.
\end{remark}

\begin{remark}
 For real data, this method may be problematic, as it is highly sensitive to deviations from the theoretical model. Determining modifications which result in a stable and robust methodology is an area of continuing research. Given the feedback effects in the EM algorithm, this problem can be significant, and readily appears for real data.
\end{remark}

\subsubsection{Example: Discrete observation Hawkes' process}
In discrete time, various approximations of this log-likelihood are possible. Given the previously calculated values of $\blambda$ and $r$, from (\ref{eq:lambdaDTapprox}) and (\ref{eq:filterDTapprox}),  a simple approximation is given by 
\[\sum_{i}\Big\langle \hat r_{t_i} ,  \big(\delta Y_{t_i} \log(\blambda_{t_{i-1}})\big) - \blambda_{t_{i-1}}(t_i-t_{i-1})\Big\rangle\]
which can be rapidly calculated. Maximization of this function can then be done using a variety of numerical methods.


\subsection{Underlying chain dynamics}
To estimate the underlying chain transition matrix is typically very costly, as there will be few jumps of the chain compared with the number of observations. This implies that, even with perfect observation of the underlying chain, the estimation of the transition matrix would be poor. 

One suggested method (see, for example \cite{Elliott1993}), is to use the EM algorithm, with the hidden variable being the number of transitions between states, and the occupation times of each state. This method, however, degenerates in this context, as will be seen in the following.

Let $J$ be the random matrix with entries $J_{ij}$, where $J_{ij}$ is the number of transitions from state $i$ to state $j$ over the observation period, and $J_{ii}=-\sum_j J_{ij}$ for all $i$. Let $K$ be the vector with the occupation time in each state. Then, from the structure of $X$, as $\hat r_t = E[X_t|\Y_T]$ we know
\begin{align*}
 E[K|\Y_T] &= E\Big[\int_{]0,T]} X_{t-}dt\Big|\Y_T\Big]  = \int_{]0,T]} \hat r_tdt \\
 E[J|\Y_T] &= E\Big[\int_{]0,T]} X_{t-} dX_t^\top\Big|\Y_T\Big]\\
&= E\Big[\int_{]0,T]} X_{t-}X_{t-}^\top A^\top dt + \int_{]0,T]} X_{t-} dM_t^\top\Big|\Y_T\Big]\\
&= \int_{]0,T]} E[X_{t-}X_{t-}^\top|\Y_T] A^\top dt=  \diag\Big(\int_{]0,T]} r_tdt \Big) A^\top.
\end{align*}
Therefore, given an initial estimate $\hat A$ of $A$, we have
\[E[J|\Y_T]\approx \diag\Big(\int_{]0,T]} \hat r_tdt \Big) \hat A^\top = \diag( E[K|\Y_T]) \hat A^\top.\]
The maximum likelihood estimate of $A$ given the numbers of transitions ($J$) and the occupation times ($K$) is 
$J^\top\diag(K)^{-1}$, so the EM algorithm gives the estimate
\[E[J|\Y_T]^\top\diag(E[K|\Y_T])^{-1} = \hat A \diag(E[K|\Y_T])\diag(E[K|\Y_T])^{-1}=\hat A\]
and therefore, starting from some initial estimate $\hat A_0$ and iterating the EM algorithm, we obtain the sequence of estimates $\hat A_0=\hat A_1=\hat A_2=...$  From this, it is clear that the EM algorithm applied in this way does note yield a consistent estimator of the rate matrix.

An alternative perspective, which we adopt below when modelling real data, is to regard the rate matrix in the model as a `tuning parameter' for the filter, which should be determined using expert judgement during calibration. From this perspective, one chooses a parameterized family of rate matrices, for example, 
\[\Big\{A=\epsilon\left[\begin{array}{cc} -1&  1\\  1& -1\end{array}\right]; \epsilon >0\Big\}\]
and then chooses, during calibration, the value of $\epsilon$ which gives acceptable performance of the filter in determining the underlying state. When the value of $\epsilon$ is large, the state is allowed to shift frequently, while when $\epsilon$ is slow, stronger data is needed before the filter detects a state change. The effectiveness of this choice of parameter should then be determined by considering out-of-sample performance.

\section{Simulated Numerical Results}\label{sec:simulated}

To assess the accuracy of these methods, the algorithm given by equations (\ref{eq:lambdaDTapprox}) and (\ref{eq:filterDTapprox}) was implemented in \texttt{R}. Data was simulated by first simulating the Markov chain (in continuous time) and then using a combination of the approach of Ogata \cite{Ogata1981} and a branching method (as suggested in Algorithm 1 of \cite{Moller2005}, ignoring edge effects) to simulate the Hawkes' process between jumps of the Markov chain.

Parameter values were chosen to give the following behaviour. We have a two-state chain with a few jumps over the timescale of simulation.  In the first state, the $Y$ process is noticeably self-exciting, however has a low base rate, that is, $\alpha$ is small, $\beta$ and $\gamma$ are moderately large. In the second state, the process is negligiably self exciting, but has a higher base rate, that is, $\alpha$ is moderately large, $\beta$ is small and $\gamma$ is moderately small. The parameters were then chosen so that, if the underlying state remains the same, the long run mean rate of jumps in $Y$ is the same in each state. The values chosen, for a simulation horizon of $T=1000$, were the transition matrix
\[A=\left[\begin{array}{cc} -0.01&  0.01\\  0.01& -0.01\end{array}\right]\]
and observation parameters
\begin{center}
\begin{tabular}{l ||c |c |c}
          & $\alpha$ & $\beta$ & $\gamma$\\
\hline
 State 1  & $6$      & $1$     & $10/7$\\
 State 2  & $18$     & $0.01$  & $0.1$\\
\end{tabular}
\end{center}
The average rate of jumps in each state, over the long-run without a state change, is $20= \frac{\alpha}{1-\beta/\gamma}$ per unit time. The number of jumps in each time interval of length $0.1$ was recorded. A typical sample path is shown below (Figure \ref{fig:simobs}).
\begin{figure}[ht]
\begin{center}
\includegraphics[width=0.9\textwidth]{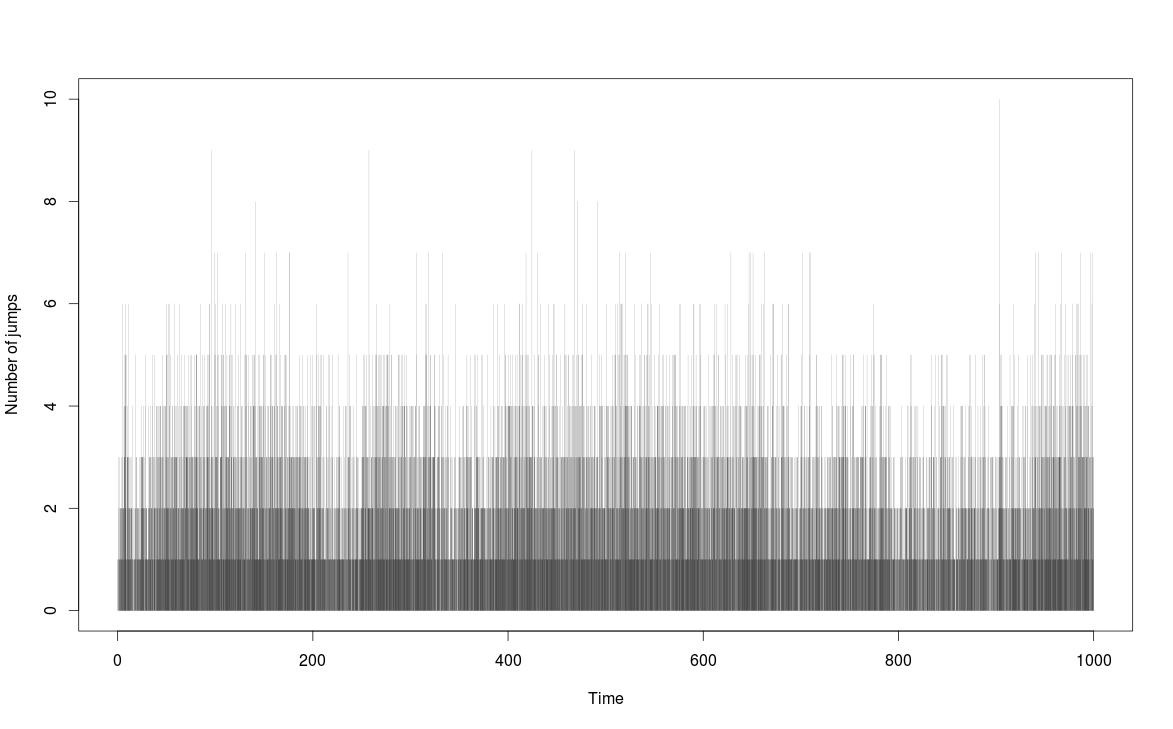}
\caption{A simulated observation path for the Markov-modulated Hawkes process}
\label{fig:simobs}
\end{center}
\end{figure}

For the above sample path, a crude initial clustering is to say that there are state changes at times approximately $50, 200, 250, 300, 400, 650$ and $950$, which can be performed by eye. This initial allocation can then be used to estimate the parameters, using the maximum likelihood estimator. The transition matrix $A$ is taken to be known. Optimisation of the likelihood was performed using the inbuilt Nelder--Mead method in the \texttt{optim} command. Given these estimates, one can apply the filter to determine the hidden state, and hence recalibrate, and repeat. The first few iterations of the parameter estimates are as follows.
\begin{center}
\begin{tabular}{l|cc|cc|cc}
 Iteration & $\alpha$ && $\beta$ && $\gamma$ &\\
\hline
1& 7.35 &20.23 &0.813& $8.6\times 10^{-08}$& 1.34& 0.207\\
2& 7.54 &19.85 &0.869& 0.0032 &  1.485 & 0.178\\
3& 7.64 &19.59 &0.891& 0.0050 &  1.538 & 0.164\\
4& 7.60 &19.45 &0.893& 0.0064 &  1.534 & 0.169
\end{tabular}
\end{center}

These estimates have some error, however the filtered and smoothed paths, using these final parameters, are given below (Figure \ref{fig:simfilt}), along with the true (unknown) state for comparison. From these paths, we can see that the filter sill has good performance with these noisy parameter values. We note that there are a couple of state changes (between $t=700$ and $t=800$) which are not detected by the smoother. This is unsurprising, as these changes were only of a short duration, and so the smoother determines that any apparent change in the rate of jumps is more likely due to random variation.
\begin{figure}[ht]
\begin{center}
\includegraphics[width=0.9\textwidth]{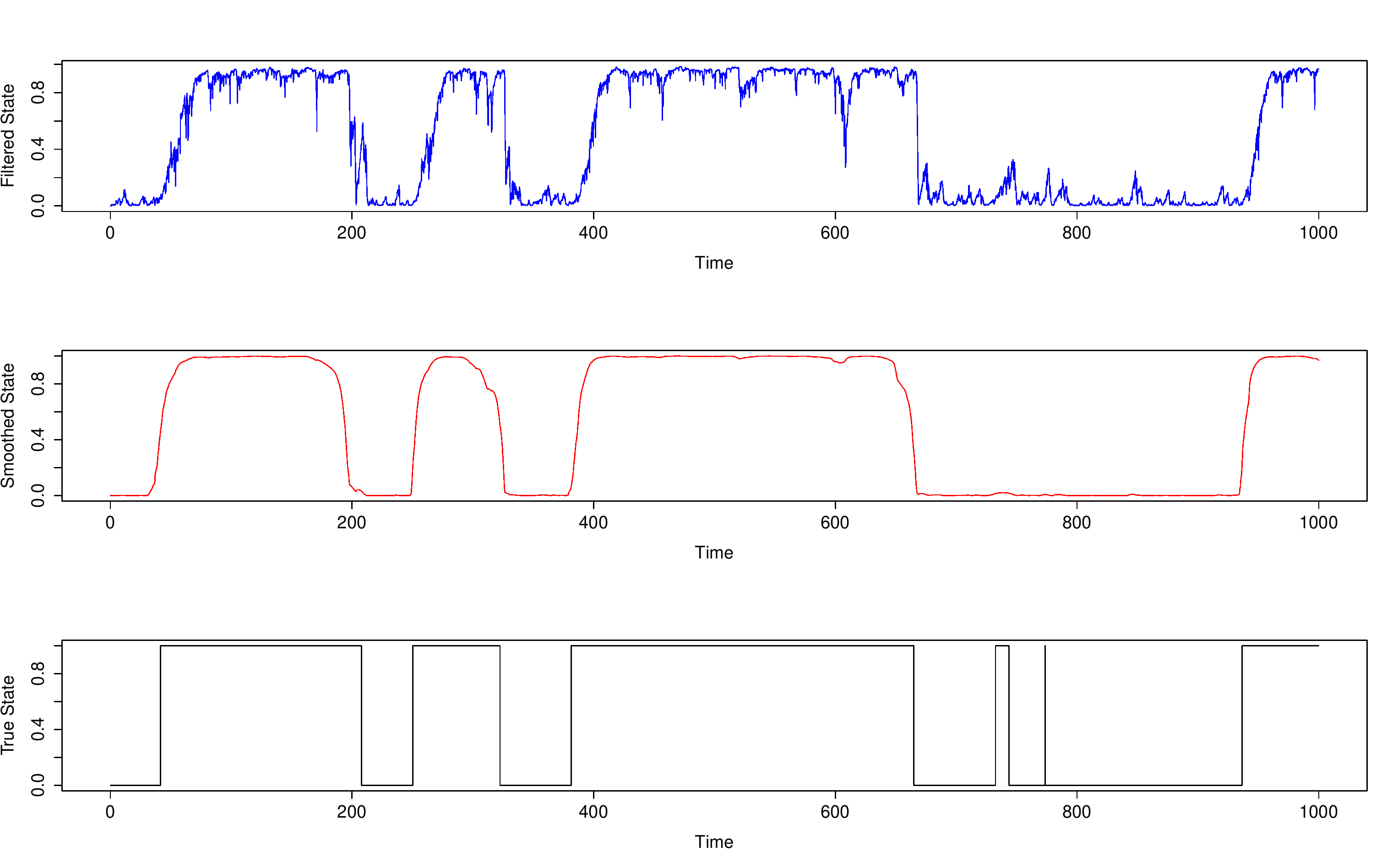}
\caption{Filtered, Smoothed and True hidden states (in particular the probability that $X_t=e_2$) for the observation path in Figure \ref{fig:simobs}.}
\label{fig:simfilt}
\end{center}
\end{figure}

We see that, for these parameter values, the filter performs well at detecting a change from State 2 to State 1 (the change is quite sharp), however is less effective at detecting changes from State 1 to State 2. The reason for this can be seen in Figure \ref{fig:simrates}, which shows the values of $\lambda(e_i, t, Y_{(\cdot)})$ for each $e_i$. When $X=e_1$ (which is the case in the middle of the plot), the difference in the potential rates in each state is quite pronounced. The low base rate of jumps means that $\lambda(e_1, t, Y_{(\cdot)})$ is typically low, with occasional large spikes. Conversely, when $X=e_2$, the rate of jumps is far more stable. Therefore $\lambda(e_1, t, Y_{(\cdot)})$ and $\lambda(e_2, t, Y_{(\cdot)})$ are relatively close, and $\lambda(e_1, t, Y_{(\cdot)})$ is less volatile than when $X=e_1$. As it is the relative difference in the rates which is important when determining the filter, we see that a real state change from State 2 to State 1 results in possible values of $\lambda$ 
becoming more distinct, and so such a change is quickly detected by the filter.
\begin{figure}[ht]
\begin{center}
\includegraphics[width=0.9\textwidth]{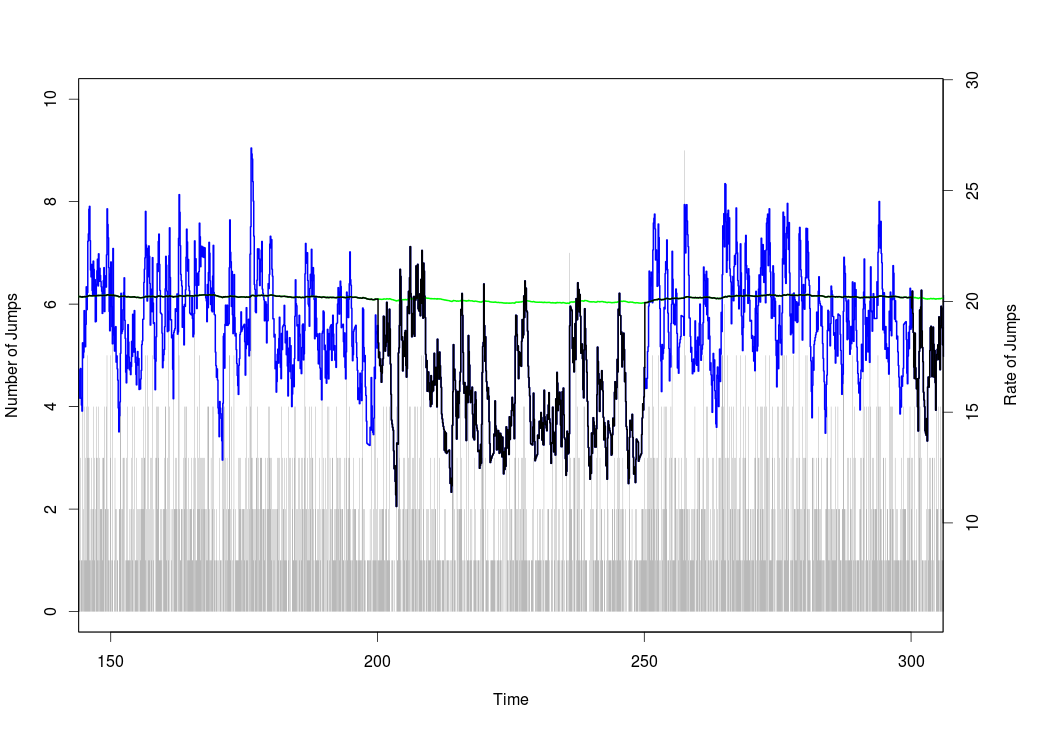}
\caption{Rates in each state for the observation path in Figure \ref{fig:simobs}, for times $t\in [150,300]$. The function $\lambda(e_1, t, Y_{(\cdot)})$, corresponding to the rate if $X$ were in State 1, is shown in blue, while $\lambda(e_2, t, Y_{(\cdot)})$ is in green, and $\lambda(X_t, t, Y_{(\cdot)})$ is in black.}
\label{fig:simrates}
\end{center}
\end{figure}

\section{Real-World Application}\label{sec:flash}

We now outline a possible application of these methods to real-world data. Given the complexities of building good models for the real world, this section should be seen as empirical support for the utility of these methods, rather than as a practical prescription for modelling. In particular, we will use realistic but ad-hoc choices of parameters, selected to demonstrate the effectiveness of the filter. We will give times in `$t=$seconds after 9:30am'.

On the 6th May 2010, markets experienced the now well known `flash-crash', in which the Dow fell 600 points (5-6\%) in the space of 5 minutes, and then rapidly recovered. A primary factor in this was  trading on the E-mini S\&P futures contract on the Chicago Mercantile Exchange Globex platform. A report \cite{UCFT2010} on this event describes the market behaviour in detail, a further study can be found in \cite{Kirilenko2011}. After a day of volatile trading, a fundamental trader issued a large sell order on the E-mini contract at 2:32pm ($t=18120$), which was executed without regard to price or time, over a period of 20 minutes. This sell order was, most likely, primarily absorbed by high frequency traders, fundamental buyers in the futures market, and cross market arbitrageurs (particularly through pairs trading with the SPY contract\footnote{The SPDR S\&P500 ETF (SPY) is an exchange traded fund replicating the S\&P500, traded on NYSE.} or fundamental securities in the S\&P500 index). As traders modified 
their positions, the combined sell 
pressure drove the prices of the E-mini 
and SPY down by 3\% in the three minutes 2:41pm--2:44pm ($t=18660 \text{ to } 18840$). This fall created a hot-potato effect, which lead to further falls until 2:45:28pm ($t=18928$), when a five second trading halt was placed on the E-mini contract. Between 2:41pm and the trading halt, the E-mini price had fallen 5\%, while the SPY price had fallen 6\%. These contracts then recovered, in a volatile way, until around 3pm ($t=19800$), when prices stabilised near their original levels. This event caused a 31.7\% increase in the S\&P500 Volatility Index (VIX).

We will consider NYSE Euronext TAQ data for the SPY contract.  We will attempt to use our filter to automatically identify the flash crash, using only the frequency of trades on SPY. The data source we will consider consists only of timestamped trades on SPY, and the times are only at a one-second precision level. We will focus our attention on the day of the crash, and on the data on trades (rather than including quotes). We will consider the number of trades, rather than the traded volume, as this should display self exciting behaviour among high frequency traders more readily.

We plot the trade data (prices and numbers of trades) for this day in  Figure \ref{fig:data}. Kirilenko et al.~\cite{Kirilenko2011} identify  the period 14:32--15:08 ($t=18120\text{ to }20280$) as the period of the crash, this is shown in red in Figure \ref{fig:data}, and a higher-resolution plot for this period is given in Figure \ref{fig:data2}. From the recorded trade prices, we can clearly see the period of the crash, and also the increase in trading over the afternoon of this day. We can also see that the data appears to contain noticeable recording errors, particularly during and following the period of the crash, with some trades being recorded at prices well away from the bulk of the market (or possibly at the wrong times). This is common in high frequency data \cite{Brownlees2006, Falkenberry2002}, however as our method only uses the number of trades, it is reasonably robust against these errors, and we perform no further data cleaning.

\begin{figure}[h]
\begin{center}
 \includegraphics[width=0.9\textwidth]{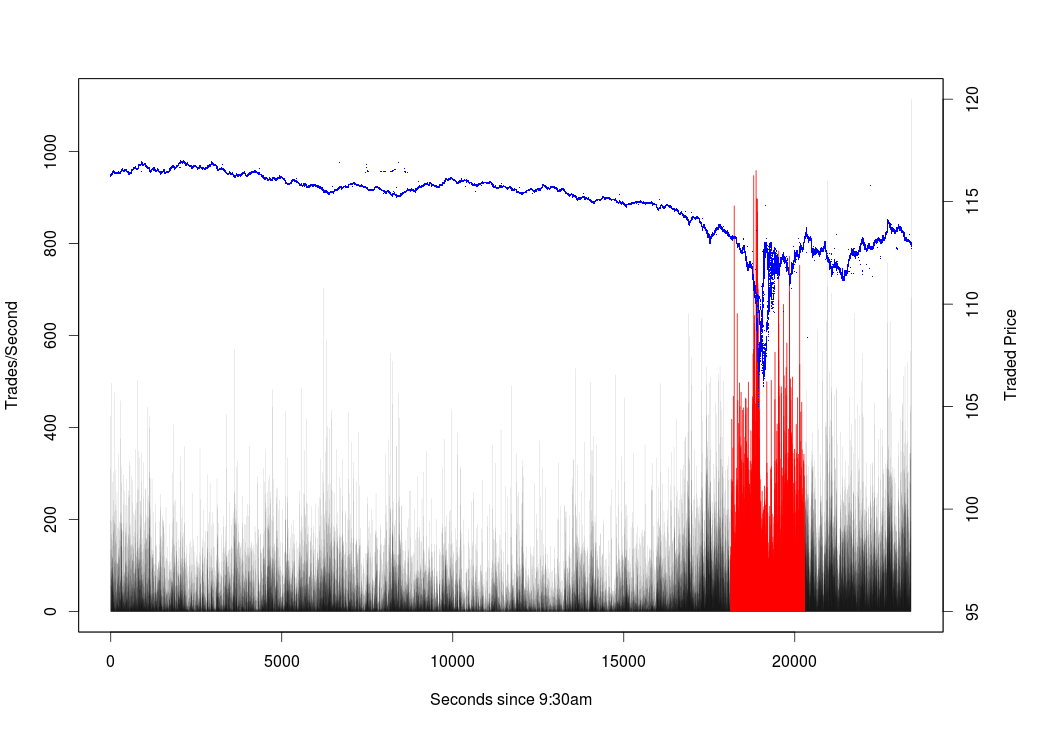}
\caption{NYSE Euronext SPY Trades 6th March 2010 (9:30--16:00), recorded prices (blue) and numbers of trades (black). The trades in the period 14:32--15:08 are shown in red.}
\label{fig:data}
\end{center}
\end{figure}
\begin{figure}[h]
\begin{center}
 \includegraphics[width=0.9\textwidth]{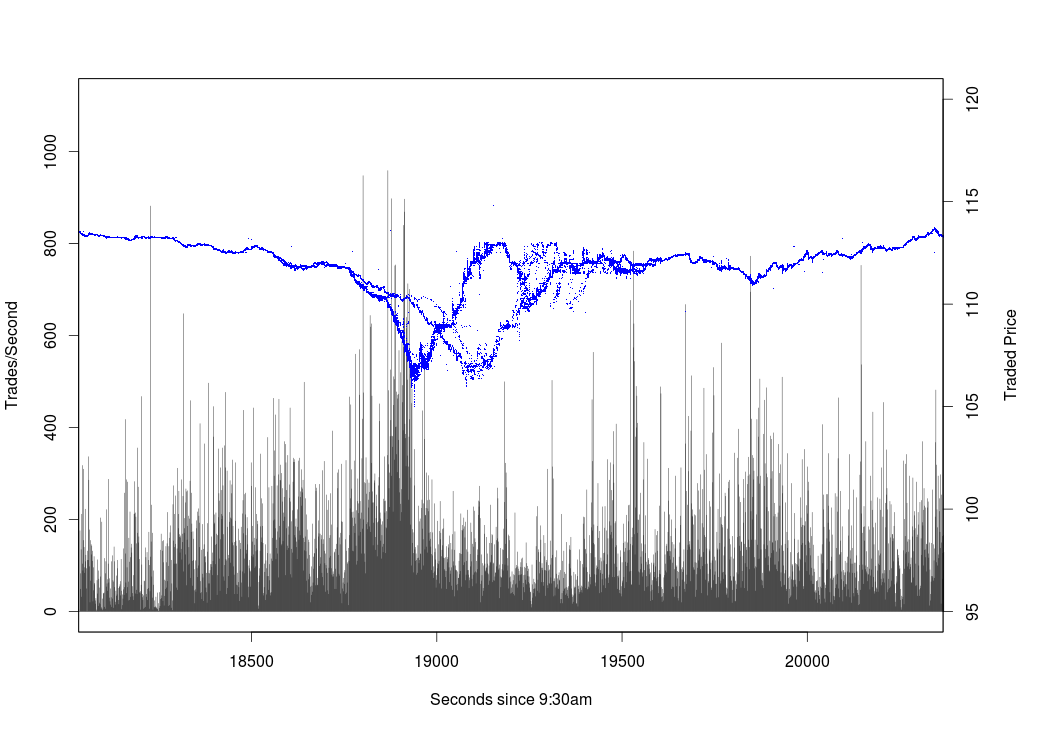}
\caption{NYSE Euronext SPY Trades 6th March 2010 (14:32--15:08), recorded prices (blue) and numbers of trades (black).}
\label{fig:data2}
\end{center}
\end{figure}

We will model the market as being in one of two states, one being a normal state $X=e_1$, the other relating to a high trading period $X=e_2$, which we hope will reflect the crash. 

For a crude approximation, we will use the period from 13:43:20 to 16:00 as the period affected by the crash, while the period from 09:30 to 13:43:20 will be used to calibrate usual trading conditions. (This time is taken as it is 15000 seconds after 9:30am.)  We will calibrate our model using a single pass with this grouping, rather than applying the EM algorithm, as the naive EM algorithm for self-exciting processes is highly sensitive to deviation from the model. This is highlighted by the fact that the numbers of trades per second appear to be significantly overdispersed relative to a Poisson distribution, to an extent which cannot be corrected for by using a self-exciting model on the one-second timescale (which is the smallest resolution of our model). Therefore, as an ad-hoc correction we will also simply divide the number of trades per second by 25, as this leads to a working model. This rescaled process we denote $Y$.

We assume a Markov-modulated-Hawkes' process of the type considered earlier, with an additional polynomial term, which allows slightly more flexibility in the effect of jumps. That is, we consider $Y$ as well approximated by a pure jump process with rate
\[\Big(\langle \alpha, X_t\rangle + \langle \beta, X_t\rangle\int_{[0,t[} e^{-\langle \gamma, X_t\rangle (t-s)} dY_s\Big)^{\langle \zeta, X_t\rangle}.\]
Discretizing the integral on our one-second scale, and using our initial grouping, we obtain maximum likelihood estimates 
\begin{center}
\begin{tabular}{c||c|c|c|c}
State& $\alpha$ & $\beta$ & $\gamma$ & $\zeta$\\
\hline
$X=e_1$& 1.0014241&  1.0416288 &0.9996515 &0.8987939\\
$X=e_2$& 0.5222101&  1.7255265 & 0.5095281&0.6821643\\
\end{tabular}
\end{center}

The rate matrix $A$ of the underlying Markov chain can be chosen as a tuning parameter for our model. Taking $A$ to be of the form
\[A= \epsilon \left[\begin{array}{cc} -1 &1 \\ 1& -1\end{array}\right],\]
we find that $\epsilon = 10^{-7}$ yields good performance, and corresponds to a self-excitation episode every 115 days. These choices of parameters are clearly ad-hoc, and are made to demonstrate the potential performance of the filter. 

The result of the filter and smoother can be seen in Figures \ref{fig:filter1} and \ref{fig:filter2}.

\begin{figure}[ht]
\begin{center}
 \includegraphics[width=0.9\textwidth]{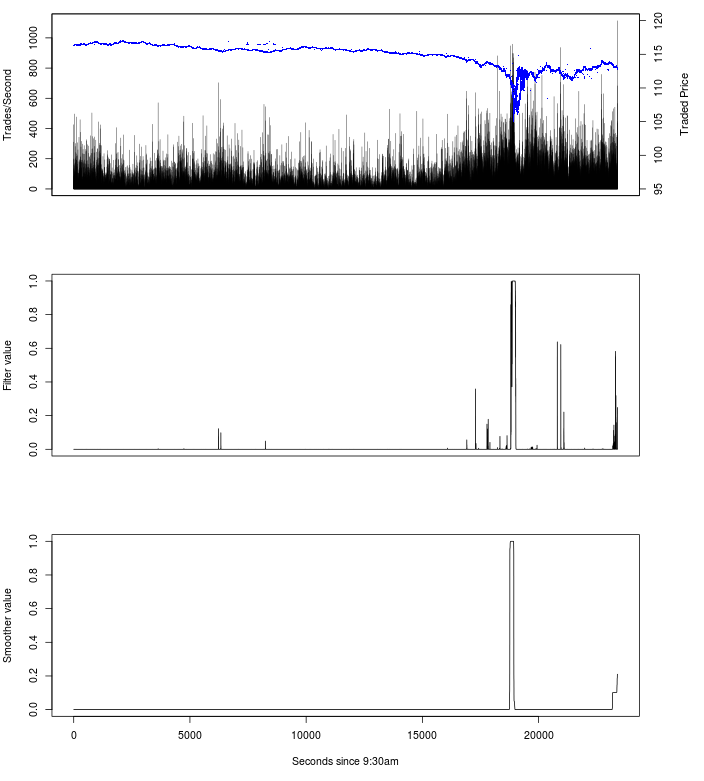}
\caption{SPY Trades and prices 6 May 2010, with filtered and smoothed state estimates.}
\label{fig:filter1}
\end{center}
\end{figure}
\begin{figure}[ht]
\begin{center}
 \includegraphics[width=0.9\textwidth]{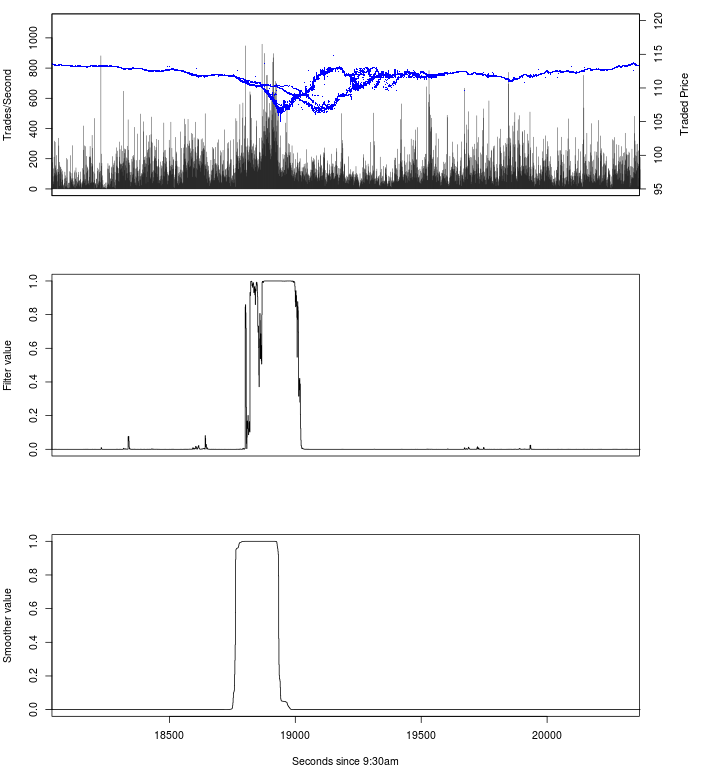}
\caption{SPY Trades and prices 6 May 2010 14:32--15:08, with filtered and smoothed state estimates.}
\label{fig:filter2}
\end{center}
\end{figure}

As can be seen, the smoother does not detect any state changes, apart from the crash itself, during the day of the crash. A few anomalous points are observed where the filter suggests that a state change is possible, however these are all quickly resolved in the following seconds (apart from those at the end of the day). Considering Figure \ref{fig:filter2} more closely, we see that, as expected, the filter detects the crash slightly after the smoother, but still before the precipitous decline in the price. The filter reverts to the usual state half-way through the price recovery. The smoother succeeds in detecting the crash, from the beginning of the price decline, and reverts to the usual state at the beginning of the price recovery.  Recall that this calculation is done without observation of the price, as the filter and smoother are based purely on the number of trades. 

\section{Conclusion}

We see that these filtering methods have significant potential to detect anomalous trading behaviour, allowing for the presence of self-excitation in both the usual and anomalous states. A key advantage of this methodology is that it is highly flexible in the choice of self-excitation structure, and, as the filter and smoother equations are available in a closed form, can be implemented at very high speed. This is of particular importance given the increase in automatic high-frequency trading, and the potential problems that this can create. 

Future developments of this work will necessitate the construction of more realistic and flexible models for the trade data. In particular, our ad-hoc correction for overdispersion (dividing the number of trades by 25) is unsatisfying, and development of multilevel self-exciting models, which can directly manage the `spiky' nature of the data set, would be preferable. The development of robust statistical methods is also signficant, particularly given the use of the EM algorithm and the errors common to high-frequency data.

\bibliographystyle{plain}
\bibliography{../RiskPapers/General}
\end{document}